\documentclass[a4paper,onecolumn,11pt,accepted=2026-07-08]{quantumarticle}
\pdfoutput=1
\usepackage[utf8]{inputenc}
\usepackage[english]{babel}
\usepackage[T1]{fontenc}
\usepackage{amsmath}
\usepackage{amsthm}
\usepackage{hyperref}
\usepackage{enumerate}
\usepackage{tikz}
\usepackage{lipsum}
\usepackage{amssymb}
\usepackage{doi}
\usepackage{braket}
\usetikzlibrary{decorations.pathreplacing}
\newtheorem{theorem}{Theorem}

\newtheorem{fact}[theorem]{Fact}
\newtheorem{claim}[theorem]{Claim}

\newtheorem{observation}[theorem]{Observation}
\newtheorem{corollary}[theorem]{Corollary}
\newtheorem{lemma}[theorem]{Lemma}
\newtheorem{definition}[theorem]{Definition}
\newtheorem{openproblem}[theorem]{Problem}

\newtheorem*{theorem*}{Theorem}
\newtheorem*{definition*}{Definition}
\newtheorem*{lemma*}{Lemma}
\newtheorem*{corollary*}{Corollary}
\newcommand{\norm}[1]{\left\| #1 \right\|}
\newcommand{\ketbra}[2]{{{\ket{#1}}\hspace{-2pt}{\bra{#2}}}}

\begin{document}
\def\bH {\mathbb{H}}
\def\bR {\mathbb{R}}
\def\bZ {\mathbb{Z}}
\def\CP {\mathbb{C P}}

\def\bC {\mathbb{C}}

\def\bI {\mathbb{I}}
\def\bx {\mathbf{x}}

\newcommand{\Tr}{{\rm{Tr~}}}
\newcommand{\tr}{{\rm{tr~}}}

\def\N {\mathcal{N}}
\def\k{\ket}
\title{Stabilizer Ranks, Barnes Wall Lattices and Magic Monotones}
\author{Amolak Ratan Kalra}
\affiliation{Institute for Quantum Computing, University of Waterloo, Waterloo, Ontario, Canada}
\affiliation{David R. Cheriton School of Computer Science, University of Waterloo, Waterloo, Ontario, Canada}
\affiliation{Perimeter Institute for Theoretical Physics, Waterloo, Ontario, Canada}
\author{Pulkit Sinha}
\affiliation{Institute for Quantum Computing, University of Waterloo, Waterloo, Ontario, Canada}
\affiliation{David R. Cheriton School of Computer Science, University of Waterloo, Waterloo, Ontario, Canada}

\maketitle

\begin{abstract}
In 2024, Kliuchnikov and Schönnenbeck showed a connection between the Barnes Wall lattices, stabilizer states and Clifford operations. In this work, we study their results and relate them to the problem of lower bounding stabilizer ranks. We show the first quantitative lower bound on stabilizer fidelity as a function of stabilizer ranks, which reproduces the linear-by-log lower bound for $\chi_{\delta}({\k{H}^{\otimes n}})$, i.e., on the approximate stabilizer rank of $\ket{H}^{\otimes n}$. In fact, we show that the lower bound holds even when the fidelity between the approximation and $\ket{H}^{\otimes n}$ is exponentially small, which is currently the best lower bound in this regime. This is shown via an intermediate result between stabilizer fidelity and stabilizer rank , which also gives us an $\Omega(\log n /\log\log n)$ lower bound on the stabilizer rank of pseudorandom quantum states.

\noindent Next, we define a new magic monotone for pure states, the \emph{Barnes Wall norm}, and its corresponding approximate variant. We upper bound these monotones by the $CS$-count of state preparation and also by the stabilizer ranks. In particular, the upper bound given by the $CS$-count is tight, in the sense that we exhibit states that achieve the bound.

\noindent Apart from these results, we give a Fidelity Amplification algorithm, which provides a trade-off between approximation error and the stabilizer rank. As a corollary, it gives us a way to compose approximate stabilizer decompositions into approximate decompositions of their tensor products. In particular this recovers the best known approximation for ${\k{H}^{\otimes n}}$. Furthermore we show that this best known approximation is effectively a Barnes Wall lattice approximation. It can also be shown that the Barnes wall norm of this approximation asymptotically matches the upper bound we present.

\noindent Finally, we provide an alternative, elementary proof of the existence and density of product states with maximal stabilizer ranks, which was first proven by Lovitz and Steffan (2022), where they used results from algebraic geometry.
\end{abstract}
\section{Introduction}
Understanding the boundary between classical and quantum computation is a problem of vital importance in the field of quantum information. One way to approach this problem is via the lens of classical simulation of quantum circuits. In this arena, we attempt to quantify the cost of simulating different families of quantum circuits by giving explicit classical algorithms to simulate them and then study how the efficiency of these algorithms scales with the number of input qubits.
The first landmark result in this approach was given by Gottesman and Knill \cite{gottesman1998heisenbergrepresentationquantumcomputers}, which showed that arbitrary quantum circuits using stabilizer states and measurements can be classically simulated efficiently (see also \cite{Aaronson_2004}). While these stabilizer circuits are not universal, they can be made to perform universal quantum computation by supplementing them with special non-stabilizer states called magic states.

In \cite{smith}, Bravyi, Smith and Smolin described a new algorithm to simulate these universal quantum circuits. They decomposed the magic states as linear combinations of a small number of stabilizer states, and then simulated the stabilizer circuits on this combination instead. The minimum number of such stabilizer states needed for a given magic state $\ket\psi$ is referred to as the \emph{stabilizer rank} of $\ket{\psi}$, denoted as $\chi(\k\psi)$. For such simulation algorithms, $\chi({\k\psi})$ is effectively the most important metric of efficiency or computational complexity. In particular, if we are able to efficiently find polynomially sized decompositions for magic states like $\ket{H}^{\otimes n}$, then this has various complexity-theoretic implications such as $BQP=BPP$ and even $P=NP$ \cite{Peleg2022lowerbounds}. Also, instead of finding an exact decomposition in terms of stabilizers, we can also hope to approximate the magic states with linear combination of stabilizer states instead. The corresponding rank in this case is called the $\delta$-approximate stabilizer rank, denoted as $\chi_{\delta}(\k\psi)$.

The problem of finding optimal bounds for $\chi(\ket{\psi})$ and $\chi_{\delta}(\k\psi)$ for tensor powers of magic states, termed the \emph{stabilizer rank} problem, has received a lot of attention, with various works chipping away at both the upper and lower bounds, with most results focusing on the ranks of $\ket{H}^{\otimes n}$ or (equivalently) $\ket{T}^{\otimes n}$. Here, $\ket H=\cos \frac \pi 8 \ket 0+\sin\frac \pi 8 \ket 1$, and $\ket T=\frac 1 {\sqrt 2}(\ket 
0 +e^{i\pi/4}\ket 1)$.

For $\chi(\k H^{\otimes n})$, the first lower bound of $\Omega(\sqrt n)$ was given by Bravyi, Smith and Smolin in \cite{smith}, with a corresponding upper bound of $O(2^{n/2})$. In a series of works, \cite{smith,Bravyi2019simulationofquantum, Bravyi_2016,Peleg2022lowerbounds,Saeed,Lovitz2022newtechniques,Qassim2021improvedupperbounds}, these lower and upper bounds were improved to $\Omega(\frac{n^2}{\mathrm{polylog}(n)})$ and $O(2^{0.4n})$ respectively. 

For $\chi_{\delta}(\k H^{\otimes n})$, the best known lower and upper bounds are $\Omega(\frac{n^2}{\mathrm{polylog}(n)})$ \cite{Saeed} and $O(2^{0.23n})$\cite{Bravyi2019simulationofquantum} for constant $\delta$.

In order to study the stabilizer rank, various other measures of stabilizer complexities have been introduced, including but not limited to: stabilizer fidelity, stabilizer extent and stabilizer alignment (see \cite{Bravyi2019simulationofquantum}), with emphasis on relating these to each other and as well as to stabilizer ranks \cite{mehraban2024improved,Bravyi2019simulationofquantum, Heimendahl2021stabilizerextentis}. In particular, we know that tensor products of magic states, such as $\ket{H}^{\otimes n}$, have exponentially large stabilizer extent, and exponentially small stabilizer fidelity.

\subsection{Overview of the Main Results}
In a recent work, Kliuchnikov and Schönnenbeck \cite{kliuchnikov2024stabilizer} demonstrated that for $n$ qubit Barnes Wall lattices, the automorphism group of the lattice, as well as the set of minimum length vectors, are upto phase exactly the Clifford unitaries and stabilizer states respectively. In our work, we study the problem of bounding stabilizer ranks and other notions of stabilizer complexities from the perspective of these recent results. Till now none of the techniques developed to study stabilizer ranks make use of this particular connection. Another reason to pursue this connections is that similar number theoretic techniques have been used to give synthesis algorithms for single qubits with close to optimal $T$-count \cite{kmm,selinger}.

We first use such a lattice interpretation in Section \ref{sec_linearlowerbound} to prove the first quantitative lower bound on stabilizer fidelity based on the stabilizer rank. Such a bound was shown to exist by Mehraban and Tahmasbi \cite{mehraban2024improved}, but the result was only existential and did not give a concrete bound. We do this via an intermediate result between the stabilizer rank and stabilizer extent:

\begin{theorem}[Upper bound on stabilizer extent]
    \label{thm_extent_vs_rank}
    Let $\ket{\phi}$ be a pure state such that $\ket\phi=\sum_{j=1}^k c_i\ket {s_j}$, where $\ket{s_j}$'s are all linearly independent stabilizer states, $c_j\in \mathbb C$. Then, we have that $\norm{\mathbf c}_1\leq \sqrt e\cdot (2k)^{\frac {2k+1}2}$. 
\end{theorem}
This theorem is established using Minkowski's theorem for lattices. At an intuitive level, it tells us that the determinant of the Gram matrix of lattice generators is large, if the minimal vector of the lattice is long. When we consider the lattice generated by the stabilizer states in the above theorem, since this determinant is large, we can show that the stabilizer extent is small.
\begin{theorem}[Lower bound on stabilizer fidelity]\label{thm_fid_vs_rank}
    Let $\ket \phi$ be a pure state with stabilizer fidelity $F_S$. For any state $\ket{\psi}$ with $\chi(\psi)=k$, we have $\frac{|\braket{\phi|\psi}|}{\sqrt{F_S}} \leq \sqrt{e} (2k)^{\frac{2k+1}{2}}$. In particular, if $\frac{|\braket{\phi|\psi}|}{\sqrt{F_S}}=2^{\Omega(n)}$, then $k=\Omega\left(\frac n {\log n}\right)$. 
\end{theorem}
 This linear-by-log lower bound holds not only for the exact stabilizer rank for states with exponentially small stabilizer fidelity, such as $\ket{H}^{\otimes n}$, but it holds even when the approximate stabilizer decomposition has only some exponentially small inner product with the said states.
 \begin{corollary}[Linear-over-log lower bound]
     For any state $\ket{\psi}$ with $F(\ket\psi,\ket{H}^{\otimes n})\geq c^n$, for any constant $c>\cos^2(\frac \pi 8)$, we have that $\chi(\k\psi)=\Omega(\frac n{\log n})$.
 \end{corollary}
 Our lower bounds are the best in this regime, as none of the other results, including the quadratic lower bound in \cite{Saeed}, provide any non trivial lower bounds for this. Moreover, as observed in \cite{mehraban2024improved}, such a relation between stabilizer rank and fidelity helps imply a lower bound on stabilizer ranks for pseudorandom states using Theorem 1 from \cite{Grewal2022}.
\begin{corollary}
    Given polynomially many copies of $\ket{\psi}$, we can distinguish, with high probability, the following two cases using a polynomial sized quantum circuit:
    \begin{enumerate}
        \item $\chi(\ket\psi)=O\left(\frac{\log n}{\log \log n}\right)$
        \item $\ket \psi$ is a Haar random unit vector.
    \end{enumerate}
\end{corollary}

This improves upon the $\omega(1)$ lower bound provided in \cite{mehraban2024improved}. Note that the only observation that we used for this result was that the size of the minimal vector was a constant, i.e, we did not use the fact that the minimum is attained by the stabilizer states. To capture this notion of minimality of stabilizer states, in Section \ref{sec_magic}, we defined a new measure on states in $\mathbb C^{2^n}$, which we call the Barnes Wall norm.

\begin{definition*}
The Barnes Wall norm of a state $\ket{\psi}$, denoted $\N(\ket\psi)$, is the length of the smallest vector on the Barnes Wall lattice proportional to $\ket \psi$.
\end{definition*}
We also define the corresponding approximate variant of this measure, denoted as $\N_\delta(\k\psi)$. This norm enables us to argue about the stabilizer rank of magic states using properties exclusive to the Barnes Wall lattice in Theorem \ref{thm_monotone_vs_rank}. In particular, we show that if a state has high approximate Barnes Wall norm, then it also has high approximate stabilizer rank. It also turns out that these norms satisfy many of the properties that we expect from magic monotones:

\begin{theorem}[Magic monotone properties]~\label{thm_magic_props} The following properties hold for $\N$ and $\N_{\delta}$:
\begin{enumerate}[(i)]
\item $\N$ and $\N_{\delta}$ is invariant under the action of the Clifford Group.
\item $\N(\k\phi)\geq1$, with equality only for stabilizer states.
\item When $\N(\ket\phi),~\N(\ket{\psi})<\infty$, we have $\N(\k{\phi}\otimes{\k\psi})=\N(\k\phi)\otimes \N(\k\psi)$. In particular, $\N(\k\phi^{\otimes n})=\N(\k\phi)^{\otimes n}$.
\item For all non-stabilizer $ \k\phi$,  $\N(\k\phi^{\otimes n})=2^{\Omega(n)}$.
\item $\N$ is non-increasing under uniform probability $1$-qubit Pauli measurements (post-selecting on any outcomes). If it indeed does decrease, it has to divide the original value. 
\item If the measurement does not have uniform probability, then we have the weaker constraint that $\N(\ket\phi)\geq \mathbb E[\N(\ket{\psi})]$, where $\ket\psi$ is the random post-measurement state. 
\item $\N(\k\phi)$ is either $\infty$ or $\N(\k{\phi})$ is of the form $\frac{c}{{(1+i)}^{n}}$ where $c\in \mathbb{Z}[i]$.
\end{enumerate}
\end{theorem}
The first two properties follow directly from the fact that stabilizer states are minimal vectors of the Barnes Wall lattice, and the Clifford group is the automorphism group of the lattice. The other properties of this norm can be shown using the dual of the Barnes Wall lattice. The dual is first used to give a characterization of finding the shortest vector on the lattice proportional to a given vector. This shortest vector for a given state is exactly the vector we use to measure the Barnes Wall norm $\mathcal{N}(\ket{\phi})$. Furthermore the divisibility properties of this norm can also be used to bound the $CS$ count:

\begin{theorem}\label{thm_monotone_CS_count}
Let $\ket{\phi}$ be an $n$-qubit state which can be prepared exactly with $m$ $CS$ gates and Clifford operations without any intermediate measurement then:
\[
\mathcal{N}(\k\phi)\,~|~\, \mathcal{N}(\k{CS})^{m}=2^m
\]
If it does need intermediate measurements then we have the weaker result of \[\mathcal{N}(\k\phi)\leq 2^m \]
\end{theorem}
Here $\ket{CS}$ is the magic state $\ket{CS}=\frac 1 2(\ket {00}+\ket{01}+\ket{10}+i\ket{11})$.
We also relate $\N_\delta(\ket{\psi})$ to the approximate stabilizer rank of $\ket\psi$, using a lattice approximation lemma:

\begin{theorem}\label{thm_monotone_vs_rank}
For any $\ket{\psi}\in \mathbb C^{2^n}$ and $\delta,\delta_0>0$, \[\N_{\delta+\delta_0}(\ket {\phi})\leq\frac {2\chi_\delta(\ket{\phi})}{\delta_0^2}.\] 
\end{theorem}

It is immediately clear that approximating magic states with stabilizer decomposition is at most as hard as providing exact decompositions. However, it is not immediate that approximation doesn't suddenly becomes hard after a certain threshold of error. That is, it is a possibility that there are some approximation constants $\delta_1,\delta_2\in (0,1)$ such that $\chi_{\delta_1}(\ket{H}^{\otimes n})$ is polynomial in $n$, while $\chi_{\delta_2}(\ket{H}^{\otimes n})$ is exponential. In Section \ref{sec_erroramplification}, we show that such scenarios are not possible, via our Fidelity Amplification result. For this, define the \emph{relative error} of pure state $\ket{\phi}$ with respect to $\ket{\psi}$ to be $\frac{1-F(\ket\phi,\ket\psi)}{F(\ket{\phi},\ket{\psi})}$, where $F(\ket{\phi}, \ket{\psi})$ is the fidelity of $\ket{\phi}, \ket{\psi}$.

\begin{theorem}[Fidelity amplification for $\ket{H}^{\otimes n}$]\label{thm_error_amp}
    Given a stabilizer decomposition of $\ket {H}^{\otimes n}$ with relative error $\epsilon$ and rank $k$, there exists another stabilizer decomposition with rank $O(\alpha k)$ and relative error $\frac {\epsilon}{\alpha}$. 
\end{theorem}

In particular, for proving super-polynomial lower bounds for the approximate rank, this result shows that it suffices to show those for decompositions with polynomially small error. Further, this also says that any such lower bound will also apply to decompositions with polynomially high relative error, i.e, with polynomially small fidelity. We prove this result by starting with a stabilizer decomposition $\ket{\psi}$ of $\ket{H}^{\otimes n}$, and looking at its components along this direction and in the perpendicular direction. On applying an $H$ gate with probability $\frac{1}{2}$ independently to each qubit of $\ket{\psi}$ to get a distribution over states $\ket{\phi}$, the expected relative error of $\frac{\ket{\psi} + \ket{\phi}}{2}$ is half of its original value. As such, there is a stabilizer decomposition of rank $2k$ with half the relative error. This result also lets us compose good approximate decompositions together to higher tensor powers:

\begin{corollary}\label{cor_decomp_compose}
    Given a $\delta$ relative error rank $k$ stabilizer decomposition for $\ket {H}^{\otimes m}$, and some $\alpha>1$, there is a rank $  O((k(1+\delta))^n)$ constant relative error stabilizer decomposition for $\ket{H}^{\otimes mn}$.
\end{corollary}

  Applying this result to $\ket{0}^{\otimes n}$ retrieves the best known approximate stabilizer decomposition for $\ket{H}^{\otimes n}$ with rank $O(2^{0.23n})$ first proved in \cite{Bravyi_2016}. In Section \ref{sec_latticeapprox} we further show that the best known approximations for $\ket{H}^{\otimes n}$ are effectively lattice approximations which asymptotically match the bound in Theorem \ref{thm_monotone_vs_rank}.

Finally, in Section \ref{sec_product_state_exist} we also provide an alternate proof for the density of product states with stabilizer rank exactly $2^n$
\begin{theorem}\label{thm_product_state_dense}
    Product states with stabilizer rank exactly $2^n$ form a dense and open subset of the set of product states.
\end{theorem}
This result had earlier been shown by Lovitz and Steffan \cite{Lovitz2022newtechniques}, where they had used techniques from algebraic geometry. In our work, we only use the vector space and metric space structure of subspaces spanned by stabilizer states to obtain the result, making these techniques more amenable to being extended to results about approximate stabilizer ranks. Note that any lower bound on the approximate stabilizer ranks of arbitrary product states implies the same lower bound (upto log factors) on $\ket{H}^{\otimes n}$, as we can approximate any product state with $O(n\cdot \mathrm{polylog}(n))$ many $T$-gates.
\section{Preliminaries}
\label{sec_prelims}
In this section we will review some background material. We begin with the recalling some basic definitions from the stabilizer rank literature see \cite{Bravyi2019simulationofquantum,Lovitz2022newtechniques,smith} for more details.
\begin{definition}[Also see \cite{Bravyi2019simulationofquantum}]
Consider a pure state $\k \psi$, then exact stabilizer rank $\chi(\ket{\psi})$ is defined to be the smallest integer $k$ for which a given state $\k\psi$ can be written as follows:
\[
\k\psi=\sum_{i=1}^{k}\alpha_{i}\k{s_{i}}
\]
where $s_{i}$ are stabilizer states and $\alpha_{i}\in \mathbb{C}$.
\end{definition}

\begin{definition}[Also see\cite{PhysRevLett.116.250501,Bravyi2019simulationofquantum}] 
The approximate stabilizer rank $\chi_{\delta}(\k\psi)$ of a state $\k \psi$ is the smallest integer $k$ such that there is pure state $\k\sigma$ with $\chi(\ket\sigma)=k$ and $\norm{\ket{\psi}-\ket{\sigma}}_2\leq \delta $
\end{definition}
\begin{definition}[Also see \cite{Bravyi2019simulationofquantum}]
The stabilizer fidelity $F_S(\k\psi)$ of a state $\psi$ is defined as follows:
\[
F_S(\k\psi)=\max_{\phi \in S_{n} }|\langle{\phi}\k{\psi}|^{2}
\]
where $S_{n}$ is the set of $n$-qubit normalized stabilizer states.
\end{definition}

\begin{definition}
The $n$-qubit Clifford group is generated by the following matrix elements:
\[
S=\begin{pmatrix}
1 & 0\\
0 & i\\
\end{pmatrix}~~~H=\frac{1}{\sqrt{2}}\begin{pmatrix}
1 & 1\\
1 & -1\\
\end{pmatrix}~~~CX=\begin{pmatrix}
1 & 0 & 0 & 0\\
0 & 1 & 0 & 0\\
0 & 0 & 0 & 1\\
0 & 0  & 1 & 0\\
\end{pmatrix}.
\]

\end{definition}
\begin{definition}
 The set of $n$-qubit \emph{stabilizer states} is $\{C\ket{0^n}:\text{ $C$ is an $n$ qubit Clifford Unitary}\}$.   
\end{definition}
Later in the paper we will also require the following magic state:
\[
\k{H}=\cos\left(\frac{\pi}{8}\right)\ket{0}+\sin\left(\frac{\pi}{8}\right)\k{1}.
\]
This is an eigenvector of the $H$ gate with eigenvalue 1. 
Recall the following result:
\begin{fact}
It was shown in \cite{Bravyi2019simulationofquantum} that $F_S(\k {H}^{\otimes n})$ = $\cos^{2n}(\frac{\pi}{8})$.
\end{fact}

\subsection{Arithmetic for $\bZ[i]$}

Recall that $\mathbb{Q}(i)$ and $\bZ[i]$ are defined as follows:
\[\mathbb{Q}(i):=\{a+bi~|~a,b\in \mathbb{Q}\}\]
\[\mathbb{Z}[i]:=\{a+bi~|~a,b\in \mathbb{Z}\}\]

$\mathbb Q(i)$ is a field and $\bZ[i]$ is a ring which is norm-euclidean. This essentially means that one can define a variation of Euclid's algorithm for computing the greatest common divisor in this ring. To see this we briefly review some basic properties of $\bZ[i]$ (see \cite{gauss} for more details). The following properties are simply generalizations of the corresponding elementary properties of $\mathbb Z$.
\begin{definition}
    For $\alpha, \beta \in \bZ[i]$, $\beta$ is said to divide $\alpha$ iff there is some $c\in \mathbb Z[i] $ such that $c\beta=\alpha$. We denote this as $\beta \, |\,\alpha$.
\end{definition}
In particular, any $k\in \bZ$ divides $\alpha=a+bi$ iff $k|a$ and $k|b$ in $\bZ$. 

 One can perform Euclid's division algorithm in $\bZ[i]$: for every $\alpha,\beta  \in \bZ[i]$ with $\beta \ne 0$, we have $\alpha=q\beta+r$ for some $q,r\in \mathbb Z[i]$, where  $|r|<|\beta|$. Note that $r=0\iff \beta|\alpha$.

 Also, since $|\alpha|^2\in \bZ$ for all $\alpha\in \bZ[i]$, it follows directly from the multiplicativity of $|\cdot|$ that $\pm 1,\pm i$ are the only elements with multiplicative inverses. These we refer to as $\emph{units}$.
 
\begin{definition}
    Define $\gcd(a,b)$ for $a,b\in \bZ[i]$ to be the largest (in absolute value) common divisor in $\bZ[i]$ of both $a,b$. This can be extended to $\gcd$ of arbitrarily large finite subsets of $\bZ[i]$. 
\end{definition}

Note that we can only specify the $\gcd$ upto multiplication by units, in our case means that $\gcd$ is unique only upto multiplication with units $\pm 1,\pm i$.

\begin{definition}
    A prime or an irreducible element $\alpha\in \bZ[i]$ is a number such that for no two $\beta,\gamma\in \bZ[i]$, $\alpha=\beta\gamma$, unless one of them is a unit.
\end{definition}
Note that the above definition is specific to $\bZ[i]$, and may not generalize to other rings.

An interesting thing that happens in $\mathbb Z[i]$ is that elements which are prime in $\mathbb Z$ may not be prime in $\mathbb Z[i]$. For example, $2=(1+i)(1-i)$.

\begin{fact}
    For any prime $p\in \bZ[i]$, if $p|ab$ for $a,b\in \bZ[i]$, either $p|a$ or $p|b$.
\end{fact}

\begin{fact}[Unique Factorisation] Every element can be uniquely written as a product of prime elements and units, upto a multiplication by unit for each prime element.  
\end{fact}
\begin{fact}[Bezout's lemma]
    $\gcd(a,b)$ is the smallest (norm) element of the form $pa+qb$, where $p,q\in \mathbb Z[i]$. 
\end{fact}

\subsection{Background on Barnes Wall Lattices}\label{sec_barnes_wall_prelims}

We now review the basic background on Barnes Wall (BW) lattices based on \cite{kliuchnikov2024stabilizer}, which the reader should consult for more details.
\begin{definition}
    A subset $L$ of an $\mathbb R$-vector space $V$ is said to be a $\mathbb Z$ lattice if each element of $L$ can be written as a $\mathbb Z$ linear combination of some finite subset of $V$.
\end{definition}
A simple example is the lattice $\mathbb Z^n$ inside $\mathbb R^n$. This can be generalised to the notion of a $\mathbb Z[i]$-lattice.
\begin{definition}
    A subset $L$ of a $\mathbb C$-vector space $V$ is said to be a $\mathbb Z[i]$ lattice if each element of $L$ can be written as a $\mathbb Z[i]$ linear combination of some finite subset of $V$.
\end{definition}
The set which we take $\mathbb Z$ or $\mathbb Z[i]$ linear combinations of are referred to as the generators. Note, that for any lattice $L$, we can create a \emph{sublattice} 
$L'\subseteq L$ by choosing some elements of $L$ as generators. This sublattice may even lie in a subspace of $V$. We now define the specific $\mathbb Z[i]$ lattices we will work with. Define the state $\ket{\tilde+}=\frac  1 {1+i}(\ket0+\ket1)=e^{-\frac {i\pi} 4} \ket +$.  For each $\mathbf x\in \{0,1\}^n$, define $\ket{\tilde {\mathbf x}}= \bigotimes_{i=1}^n \ket{\tilde{x}_i}$, where $\ket{\tilde x_i}=\ket 0$ if $x_i=0$, and $\ket{\tilde +}$ otherwise.

\begin{definition}[See \cite{kliuchnikov2024stabilizer,Nebe_2002,bwdefn}]
The $n$ qubit Barnes Wall lattice, denoted as $BW_n$, is generated by $\bZ[i]$-linear combinations of all $\k {\tilde{\mathbf x}}$ for $\mathbf x\in \{0,1\}^n$:
\[
BW_n=
\{\sum_{\mathbf x\in \{0,1\}^n} a_{\mathbf x}\k{\tilde{\mathbf x}}:~ a_{\mathbf x}\in \bZ[i]\}
\]
\end{definition}

Note that we have defined a scaled down version of the lattice as compared to \cite{kliuchnikov2024stabilizer}, and the results should hold with the corresponding scaling. We did so to ensure that the stabilizer states we work with have norm exactly 1.\footnote{In \cite{kliuchnikov2024stabilizer}, for $\mathbb Q(i)$, the lattice that they work with is created by $\mathbb Z[i]$ linear combinations of $(1+i)^n \ket{\tilde x}$} Furthermore, the usual presentation of the lattice is in terms of generator matrices, but the formulation here suffices for us. Also, \cite{kliuchnikov2024stabilizer} works more generally with Barnes Wall lattices defined over a larger class of cyclotomic fields, but for our case we restrict to $\mathbb Q(i)$.

The Barnes Wall lattice $BW_n$ also admits a \emph{dual lattice} $BW_n^*=(1+i)^nBW_n$, which provides an alternative characterization of $BW_n$, as follows:

\begin{fact}
     For $v\in \mathbb C^{2^n}$: \[v\in BW_n \iff a^\dagger v \in \mathbb Z[i]~\forall~ a\in BW^*_n\]
\end{fact}
It can easily be seen that the set of elements satisfying the latter condition is indeed closed under addition, however, showing that the set is exactly $BW_n$ is a non trivial exercise. Also, from the above, since $\{\ket{\tilde{\mathbf x}}\}$ for $\mathbf x\in \{0,1\}^{n}$ generate $BW_n$, we have that $\{(1+i)^n\ket{\tilde{\mathbf x}}\}$ generate the lattice $BW_n^*$.

\begin{fact}\label{fact_barnes_wall_automorphism}
The Barnes Wall lattice is mapped to itself under the group generated by following operators:
\[
S=\begin{pmatrix}
1 & 0\\
0 & i\\
\end{pmatrix}~~~e^{-i\frac\pi4}H=\frac{1}{1+i}\begin{pmatrix}
1 & 1\\
1 & -1\\
\end{pmatrix}~~~CX=\begin{pmatrix}
1 & 0 & 0 & 0\\
0 & 1 & 0 & 0\\
0 & 0 & 0 & 1\\
0 & 0  & 1 & 0\\
\end{pmatrix}
\]
Upto multiplication by phases, this is equivalent to the Clifford group.
\end{fact}
\begin{corollary}
    All stabilizer states lie on $BW_n$ upto phase.    
\end{corollary}
\begin{fact}[Minimal vectors\cite{kliuchnikov2024stabilizer}]
    For any vector $v\in BW_n$, we have $\norm{v}\geq 1$, with $\norm{v}=1$ iff $v$ is a stabilizer state upto phase.\footnote{In \cite{kliuchnikov2024stabilizer}, since they work with a scaled up lattice, the norm of the minimal vectors in that lattice is $2^{n/2}$. Moreover, in \cite{kliuchnikov2024stabilizer}, they work with a more general norm of $\Tr_{K/Q}[\norm{v}_2^2]$, where $K=\mathbb Q(\zeta)\cap \mathbb R$, for the case when the lattice is defined over $\mathbb Z[\zeta]$. For our case, since we work with $\mathbb Z[i]$, we have $K=Q$, and thus the norm reduces to just $\norm{v}^2$.}
\end{fact}

We also recall the well-known Minkowski theorem about lattices:
\begin{theorem}[Minkowski's Theorem for lattices \cite{olds2000geometry}]\label{minkowski}
Let $L$ be a $\mathbb Z$ lattice in $\mathbb R^n$ with $n$ linearly independent generators. Then there is $x\in L \setminus \{0\}$
with $\norm x_{2} \leq  \sqrt{n}\,|\det(L)|^{\frac{1}{n}}$.
\end{theorem}

Here, $\det L$ stands for the volume of \emph{fundamental parallelopiped} of the lattice, which equals the square root of the determinant of the gram matrix $G$ of the generators, which is defined as $G_{ij}=v_i^{\dagger}v_j$ for the set of generators $\{v_j\}$.

\section{Linear Lower Bound on Stabilizer Rank}
\label{sec_linearlowerbound}

Here, we give proofs for Theorem \ref{thm_extent_vs_rank} and Theorem \ref{thm_fid_vs_rank}.
\begin{theorem}\label{thm_extent_vs_rank_restated}
    Let $\ket{\phi}$ be a pure state such that $\ket\phi=\sum_{j=1}^k c_i\ket {s_j}$, where $\ket{s_j}$'s are all linearly independent stabilizer states, $c_j\in \mathbb C$. Then we have that $\norm{\mathbf c}_1\leq \sqrt e\cdot (2k)^{\frac {2k+1}2}$. 
\end{theorem}
\begin{proof}
     Consider the $\mathbb Z[i]$ lattice generated by $\ket{s_j}$'s. This forms a sub-lattice of the $n$ qubit Barnes Wall lattice. This lattice is the same as the $\mathbb Z$-lattice generated by $\{\ket {s_j}\}$ and $\{i\ket{s_j}\}$. 

    We also know that the $\mathbb Z$-lattice generated by $\{\ket {s_j}\}$ and $\{i\ket{s_j}\}$ has minimal vectors of size $1$. Let $G$ be the Gram matrix for these $2k$ vectors (here, we consider the space $\mathbb C^{d}$ these vectors lie in as $\mathbb R^{2d}$, and this change does not alter length's of vectors). By Minkowski's Theorem (Theorem \ref{minkowski}),  we have that $\mathrm {det}(G)\geq ({2k})^{-2k}$.

     In the given decomposition for $\ket{\phi}$, we rewrite each $c_j$ as $f_j+ig_j$, using which we write $\ket\phi$ as an $\mathbb R$-linear combination of $\{\ket{s_j}\}$ and $\{i\ket{s_j}\}$. Then we have that $\norm{\mathbf f}_1+\norm{\mathbf g}_1 \geq \norm{\mathbf c}_1$. Applying the power mean inequality, we have that $\norm{\mathbf f}^2_2+\norm{\mathbf g}^2_2\geq \frac {\norm{\mathbf c}_1^2}{{2k}}$. We can now write:
     \begin{equation}  
\begin{pmatrix}
\ket{s_{1}}& \hdots & \ket{s_{k}}&i\ket{s_1}& \hdots i\ket{s_k}\\
\vdots & \vdots &\vdots&\vdots&\vdots&\vdots\\
\end{pmatrix}\begin{pmatrix}
f_{1}\\
f_{2}\\
\vdots\\
f_{k}\\
g_{1}\\
g_{2}\\
\vdots\\
g_{k}\\
\end{pmatrix}=\ket \phi\label{eqn_stab_matrix_vec_coeff}
\end{equation}

Here, the left multiplicand, which we can define as $M$, is the matrix with the vectors $\ket{s_j}$ and $i\ket{s_j}$ for each $j=1$ to $k$ as column vectors written in a basis for the $\mathbb R^{2d}$ vector space they lie in. Call the right multiplicand, which is vector of coefficients $f_j$ and $g_j$ for each $j$, as $\ket v$. With the above notation, the Equation \ref{eqn_stab_matrix_vec_coeff} can be rewritten as $M\ket v=\ket \phi$.  With this, we have that:
\[\bra v M^{\dagger} M\ket v=\bra v G \ket v=1 \]

as $\braket{\phi|\phi}=1$. Let $\ket{\Tilde v}=\frac{1}{\sqrt{\braket{v|v}}}\ket v$, i.e, a unit vector proportional to $\ket v$. Plugging it into the above expression, we get that $\bra {\Tilde v} G \ket {\Tilde v}=\braket{v|v}^{-1}=\left(\norm{\mathbf f}^2_2+\norm{\mathbf g}^2_2\right)^{-1}\leq \frac{2k}{\norm{\mathbf c}_1^2}$. Let the smallest eigenvalue of $G$ be $\lambda_{2k}$. Then, we know that $\lambda_{2k}\leq \frac {2k}{\norm{\mathbf c}_1^2}$.Now, we use this to upper bound the determinant.
\begin{claim}
    $\mathrm{det}(G)\leq e \cdot \lambda_{2k}$.
\end{claim}
\begin{proof}
    If $\lambda_1,\lambda_2\dots \lambda_{2k}$ are the eigenvalues of $G$ in non increasing order, then, we have that (by the AM-GM inequality and that $\sum_{j=1}^{2k} \lambda_j=2k$)
    \[\det(G)=\prod_{j=1}^{2k}\lambda_j=\lambda _{2k}\cdot \prod_{i=1}^{2k-1} \lambda_j\leq \lambda_{2k}\cdot \left(\frac {1}{2k-1}\sum_{j=1}^{2k-1}\lambda_j\right)^{2k-1}\leq \lambda_{2k}\left(1+\frac {1}{2k-1}\right)^{2k-1}\leq e\cdot \lambda_{2k}\]
\end{proof}

Comparing the upper and lower bound for $\det(G)$, we get:

\[(2k)^{-2k}\leq \frac{2ek}{\norm{\mathbf c}_1^2} \implies \norm{\mathbf c}_1 \leq \sqrt e\cdot (2k)^{\frac {2k+1}2} \]

\end{proof}
\begin{theorem}
    Let $\ket \phi$ be a pure state with stabilizer fidelity $F_S$. Suppose there is a state $\ket {\psi}$ with stabilizer rank $k$. Then we have that $\frac{|\braket{\phi|\psi}|}{\sqrt{F_S}} \leq \sqrt{e} (2k)^{\frac{2k+1}{2}}$. In particular, if $\frac{|\braket{\phi|\psi}|}{\sqrt{F_S}}=2^{\Omega(n)}$, then $k=\Omega\left(\frac n {\log n}\right)$.
\end{theorem}
\begin{proof}
Let $\ket{\psi}$ have an exact stabilizer decomposition with $k$ linearly independent stabilizer states $\ket{s_1}\dots \ket{s_k}$. Let $V$ be the space spanned by these states, and let $P$ be the orthogonal projector onto this space. Without loss of generality we can assume $\ket{\psi}=\frac{P\ket{\phi}}{\norm{P\ket{\phi}}_2}$, as it is the vector with the largest inner product (in magnitude) with $\ket{\phi}$ in $V$. With this, we have that $\norm{P\ket{\phi}}_2 =|\braket{\psi|\phi}|$. Now, we know that $\Tr[\ketbra{\phi}{\phi}\ketbra{s_j}{s_j}]=|\braket{\phi|s_j}|^2\leq F_S$ for each $j$. Therefore,
\begin{align}
    |\braket{\psi|s_j}|^2&=\Tr[\ketbra{\psi}{\psi} \ketbra{s_j}{s_j}]\\
    &=\Tr[(P\ketbra{\phi}{\phi}P) \ketbra{s_j}{s_j}] \frac{1}{\norm{P\ket\phi}_2^2}\\
    &= \Tr[\ketbra{\phi}{\phi}(P \ketbra{s_j}{s_j}P)] \frac{1}{|\braket{\psi|\phi}|^2}\\
    &= \Tr[\ketbra{\phi}{\phi} \ketbra{s_j}{s_j})] \frac{1}{|\braket{\psi|\phi}|^2}\\
    &\leq \frac{F_S}{|\braket{\psi|\phi}|^2}
\end{align}

Now, consider the decomposition of $\ket{\psi}$ in the $\ket{s_j}$ basis as $\ket\psi= \sum_{j=1}^{k} c_j\ket{s_j}$. Multiplying by $\bra{\psi}$ on the left, we get:

\[1=\braket{\psi|\psi}=|\sum_{j=1}^{k}c_j\braket{\psi|s_j}|\leq \norm{\mathbf c}_1\cdot \sqrt{\frac{F_S}{|\braket{\psi|\phi}|^2}}\]
\[\implies \norm{\mathbf c}_1\geq \frac{|\braket{\phi|\psi}|}{\sqrt{F_S}}\]
By Theorem \ref{thm_extent_vs_rank}, we get that $\frac{|\braket{\phi|\psi}|}{\sqrt{F_S}} \leq \sqrt e \cdot (2k)^{\frac{2k+1}{2}}$. Now, if $\frac{|\braket{\phi|\psi}|}{\sqrt{F_S}}=2^{\Omega(n)}$, then $(2k+1)\log (2k)=\Omega(n)\implies k=\Omega\left(\frac n{\log n}\right)$
\end{proof}

\section{Magic Monotones and Barnes Wall Lattices}
\label{sec_magic}
In this section we will show that the norms of vectors on the Barnes Wall lattice are closely connected to magic. In particular, we will prove Theorem \ref{thm_magic_props}, Theorem \ref{thm_monotone_CS_count} and Theorem \ref{thm_monotone_vs_rank}. We will use notation defined in Section \ref{sec_barnes_wall_prelims}. We define the following quantity:
\begin{definition} \label{def_bw_norm}
    For any vector (possibly unnormalized) $\ket {\phi}$, define
    \[
\N(\ket{\phi})=\min\{||c\ket{\phi}||_{2}^{2}: c\k{\phi}\in BW_{n},~c\in \mathbb{C^{*}} \}
\]This quantity is $\infty$ if there is no such vector on $BW_n$.
\end{definition}
where $\mathbb{C^{*}}$ is the set of non-zero complex numbers. The approximate version of this quantity $N_{\delta}$ can be defined as follows:
\begin{definition}\label{def_approx_bw_norm}
    For normalized $\ket{\phi}$, define:
    \[
\N_{\delta}(\ket{\phi})=\min\{\norm{c\ket{\phi}}_2^2: c(\ket\phi+\ket{\psi})\in BW_n,\norm{\ket{\psi}}_2\leq \delta\}
\]

Extend it to all non zero vectors by $\N_{\delta}(c\ket{\phi})=\N_\delta(\ket{\phi})$ for all $c\in \mathbb C^*$.
\end{definition}

Note that the $2$-norm is equivalent to the norm specified in \cite{kliuchnikov2024stabilizer} on $\mathbb{Q}[i]$. In our results we will restrict to $\mathbb{Q}[i]$.

\subsection{Properties of $\N$ and $\N_\delta$}
\begin{definition}
    A vector $\ket{\phi}$ is proportionally minimal in $BW_n$ if $\ket{\phi}\in BW_n$ and for all $0<|c|<1$, $c\ket{\phi}\notin BW_n$. 
\end{definition}
We will simply refer to such vectors as proportionally minimal. 
\begin{lemma}
    $\N(\ket{\phi})<\infty$ if and only if the ratios of the components along any two stabilizer states lies in $\mathbb Q(i)$.
\end{lemma}
\begin{proof}
    First we prove $\implies$: WLOG assume $\ket{\phi}$ is proportionally minimal (this does not alter the ratios).
    Then by the property of $BW_n$, we have that for any stabilizer state $\ket{s}$, $(1-i)^n\braket{s|\phi}\in \mathbb Z[i]$, which suffices. For $\impliedby$, WLOG assume that $\braket{0^n|\phi}=1$.

    Then, in the basis of tensor products of $\ket0,\ket1$, all components are in $\mathbb Q(i)$. Multiply the vector by a suitable number in $\mathbb Z[i]$ so that all of these are now in $\mathbb Z[i]$ as well. Since all tensor products of $\ket0$ and $\ket 1$ lie in $BW_n$, this vector does too, and thus $\N(\phi)<\infty$.

\end{proof}

We now characterize the proportionally minimal vectors:
\begin{lemma}\label{lem_gcd_prop_min}
$\gcd_{\mathbf x\in \{0,1\}^n}\{(1-i)^{n}\braket{\tilde{\mathbf x}|\phi}\}=1$ is a necessary and sufficient condition for proportional minimality for $\ket{\phi}\in BW_{n}$.
\end{lemma}

\begin{proof}
We know that $(1+i)^n\ket{\Tilde{\mathbf x}}$ generate the dual lattice $BW_n^*$ for $BW_n$, so $(1-i)^{n}\braket{\tilde{\mathbf x}|\phi}\in \mathbb Z[i]$ $\forall \mathbf x\in \{0,1\}^n$ is a necessary and sufficient condition for $\ket{\phi}$ to lie in $BW_n$.

Now,if the $\gcd$ is bigger than 1, then we can divide by the $\gcd$ to get a smaller vector in $BW_n$.

For the other direction, suppose that the $\gcd$ is 1, and yet the some other $c\ket{\phi}$ is proportionally minimal. Since $\ket{\phi}\in BW_n$, we have that $c\in \mathbb Q(i)$. We can find coprime elements $p,q\in \mathbb Z[i]$ such that $\frac p q=c$. Then, since $c\ket{\phi}\in BW_n$, $\ket{\phi}\in BW_n$, we have 

\[\forall a,b\in \mathbb Z[i]: ac\ket{\phi}+b\ket{\phi}\in BW_n\]
\[\implies\forall a,b\in \mathbb Z[i]: \frac{ap+bq}{q}\ket{\phi}\in BW_n\]
\[\implies \frac {\gcd(p,q)}{q}\ket {\phi}=\frac 1 q \ket{\phi}\in BW_n \]
Here, we used Bezout's lemma for $\gcd$.
However, this implies that each $(1-i)^n\braket{\tilde{\mathbf x}|\phi}$ is divisible by $q$, which contradicts the assumption that the $\gcd$ is 1. This suffices for the result.
\end{proof}
\begin{theorem}\label{propertiesofN}
~ The following properties hold for $\N$ and $\N_{\delta}$:
\begin{enumerate}[(i)]
\item $\N$ and $\N_{\delta}$ is invariant under the action of the Clifford group (Also see \cite{kliuchnikov2024stabilizer}).
\item $\N(\k\phi)\geq1$, with equality only for stabilizer states.
\item When $\N(\ket\phi),\N(\ket{\psi})<\infty$, we have $\N(\k{\phi}\otimes{\k\psi})=\N(\k\phi)\N(\k\psi)$. 

In particular, $\N(\k\phi^{\otimes n})=\N(\k\phi)^{n}$
\item For all non-stabilizer $ \k\phi$,  $\N(\k\phi^{\otimes n})=2^{\Omega(n)}$.
\item $\N$ is non-increasing under uniform probability $1$-qubit Pauli measurements (post-selecting on any outcome). If it indeed does decrease, it has to divide the original value. 
\item If the measurement does not have have uniform probability, then we have the weaker constraint that $\N(\ket\phi)\geq \mathbb E[\N(\ket{\psi})]$, where $\ket\psi$ is the random post-measurement state. 
\item $\N(\k\phi)$ is either $\infty$ or $\N(\k{\phi})$ is of the form $\frac{c}{{(1+i)}^{n}}$ where $c\in \mathbb{Z}[i]$.
\end{enumerate}
\end{theorem}
\begin{proof}
(i) and (ii) follow directly from the invariance of $BW_n$ under the action of the Clifford group. \cite{kliuchnikov2024stabilizer}.

So we will begin by proving (iii). Let $\k\phi$ be an $n$-qubit state and $\k\psi$ be an $m$-qubit state.

WLOG, suppose $\ket{\phi}$ and $\ket{\psi}$ are proportionally minimal. Clearly, we have that $\N(\k\phi\otimes\k\psi)\leq\N(\k\phi)\cdot \N(\k\psi) $, as $\k\phi\otimes \k\psi\in BW_{m+n}$, since the lattice is closed under tensor products. It remains to show that this vector is proportionally minimal. We do this by the use of Lemma \ref{lem_gcd_prop_min}.  We know that
\[
\text{gcd}_{\mathbf x_1\in \{0,1\}^{n}}
\{(1-i)^{n}\braket{\tilde{\mathbf x_1}|\phi}\}=1\]
\[\text{gcd}_{\mathbf x_2\in \{0,1\}^{m}}
\{(1-i)^{m}\braket{\tilde{\mathbf x_2}|\psi}\}=1
\]
We want to prove that $\text{gcd}_{\mathbf x_1\in \{0,1\}^n, \mathbf x_2\in \{0,1\}^m}\{(1-i)^{n}\braket{\tilde{\mathbf x_1}|\phi}(1-i)^{m}\braket{\tilde{\mathbf x_2}|\psi}\}=1.$

For this suppose $\text{gcd} \neq 1$. Then there is a $p\in \bZ[i]$ irreducible such that \[p~|~(1-i)^{n}\braket{\tilde{\mathbf x}_{1}|\phi}(1-i)^{m}\braket{\tilde{\mathbf x}_{2}|\psi}\quad \forall \mathbf x_1\in \{0,1\}^n,\mathbf x_2\in \{0,1\}^m.\]
 Choose $\mathbf y_{2}\in \{0,1\}^{m}$ such that $p \nmid (1-i)^{m}\braket{\tilde{\mathbf y}_{2}|\psi}$. Now, we know that $p~|~(1-i)^{n}\braket{\tilde{\mathbf x_1}|\phi}(1-i)^{m}\braket{\tilde{\mathbf y_2}|\psi}$ for all $\tilde{x}_{1}\in \{0,1\}^{n}$. Since $p$ is irreducible, we must have that $p$ divides $(1-i)^n\braket{\tilde{\mathbf x}_{1}|\phi} \forall \tilde{\mathbf x_1}$ which contradicts the assumption that $\ket{\phi}$ was minimal. This proves (iii).

 Now, (iv) follows directly from (ii) and (iii). For (v), we use a technique from \cite{kliuchnikov2024stabilizer}. For any state $\ket{\phi}$, let $\ket{\phi'}$ be its proportionally minimal vector. Now, we can write it as, for some vectors $\ket v, \ket w$,
 \begin{equation} 
\ket{\phi'}=\frac{\ket0\ket v+\ket1\ket w}{1+i}\label{eqn_comp_basis_qubit_split}.
\end{equation}
 Clearly, $\ket v$ and $\ket w$ will have coefficients in $\mathbb Q(i)$. We now prove that these are indeed in $BW_{n-1}$. It suffices to prove that $\braket{{\alpha}|v}, \braket{{ \alpha}|w}\in \mathbb Z[i]$ for all elements $\alpha$ from the dual lattice for $BW_{n-1}$. 

 For each such $\alpha$, notice that $(1+i)\ket 0 \ket \alpha$ and $(1+i)\ket 1 \ket \alpha$ both lie in the dual of $BW_n$. Taking their inner products with $\ket{\phi'}$, we get $-i\braket{\alpha|v}$ and $-i\braket{\alpha|w}$ respectively, which are in $\bZ[i]$ proving that $\ket v,\ket w\in BW_{n-1}$.

 Now, if $\ket{\phi}$, on computational basis measurement on the first qubit, had equal probability for both the outcomes, then we have that $\norm{\ket v}_2^2=\norm{\ket w}_2^2=\norm{\ket{\phi'}}_2^2$. From Equation \ref{eqn_comp_basis_qubit_split}, we get that $\ket{v}$ and $\ket{w}$ are proportional to the corresponding post measurement states. So $\N(\k v)=|c|^2\norm{\k v}_{2}^{2}$ for some $c\in \mathbb{Q}[i]$. From this, we get that:
 \[\N(\ket v)\, |\quad \N(\ket{\phi})=\norm{\ket{\phi'}}_2^2=\norm{\ket{v}}_2^2 \]

 and the corresponding result for $\ket w$ as well. Since other Pauli basis measurements can be implemented by using Clifford operations, we have proven (v). For (vi), note that even when the probabilities are not equal, we still have 
 \[\N(\ket v)+\N(\ket w)\leq \norm{\ket v}_2^2+\norm{\ket{w}}_2^2=2\N(\ket{\phi})\]

 To prove (vii) we first recall that $BW^{*}_{n}=(1+i)^nBW_{n}$. This implies that for all  $v\in BW_{n}$, we have $(1+i)^n{v}\in BW^{*}_{n}$, and we get $(1-i)^{n}\braket{v|v} \in \bZ[i]$, which suffices for (vii).
\end{proof}
Note that a similar property to Theorem \ref{propertiesofN} $(v)$ holds for the dyadic monotone \cite{dyadic}.
\subsection{Connection Between $\mathcal{N}_{\delta}$ and Stabilizer Ranks}
We will first start by recalling a basic lemma for $\mathbb{Z}-$lattices:
\begin{lemma}
Let $L$ be a $\mathbb Z$-lattice  generated by $k$ unit vectors $v_1,v_2\dots v_k$ in some $\mathbb R^d$. Let $V$ be the $\mathbb R$ span of these vectors. Then,  
\[
\forall \k\phi \in V~ \exists~ \k \psi \in L ~\text{s.t}~||\k \phi-\k\psi||_{2}\leq \sqrt{k}
\]
\end{lemma}
\begin{proof}

Suppose $\ket{\phi}$ has the form

\[
\k \phi=\sum_{i}c_{i}\k v_{i}=\sum_{i}\lfloor{c_{i}}\rfloor v_i+\{c_i\}v_i.
\]

Here the $\{\}$ braces represent the fractional part of $c_{i}$, and $\lfloor{c_{i}}\rfloor$ refers to the greatest integer function. Now, define $k$ pairs of random variables as follows: for each $i\in [k]$, $(A_i,X_i)$ is defined as:

\[ (A_{i},X_{i})=
\begin{cases} 
(\lfloor{c_{i}}\rfloor,\{c_i\}), & \text{wp}~ 1-\{c_{i}\} \\
(\lfloor{c_{i}}+1\rfloor,\{c_i\}-1)\, & \text{wp}~ \{c_i\}
\end{cases}
\]
Here, we have that $A_i$ are $\in \bZ$ , $|X_i|\leq 1$ and $A_i+X_i=c_i$ always. Also, $\mathbb E[X]=0$. We now split $\ket{\phi}$ as follows:
\[
\k \phi=\sum c_i\k v_{i}=\underbrace{\sum^k_{i=1}A_{i}v_{i}}_{\text{random vector on}~L}+\underbrace{\sum_{i}X_{i}v_{i}}_{{w}}
\]

Take the random vector on $L$ as $\ket {\psi}$. It suffices to show that for some choice of $(A_i,X_i)$s, we have that $\norm{\vec w}_2^2\leq k$. This works out by simply taking the expectation:
\[
\mathbb E||\sum_{i}X_{i}v_{i}||^{2}_{2}=\mathbb E[\sum_{i}X_{i}^{2}v_{i}.v_{i}]+\underbrace{E[2\sum_{i<j}X_{i}X_{j}(v_{i}v_{j})]}_{=0}\leq \sum v_i\cdot v_i=k
\]

Since the expectation is at most $k$, there is going to be a random choice for which $\norm{\vec w}_2^2\leq k$. 

\end{proof}
\begin{corollary}
\label{cor_vec_latt_approx}

Let $L$ be a $\mathbb Z[i]$-lattice  generated by $k$ unit vectors $v_1,v_2\dots v_k$ in some $\mathbb C^d$. Let $V$ be the $\mathbb C$ span of these vectors. Then,  
\[
\forall \k\phi \in V~ \exists~ \k \psi \in L ~\text{s.t}~||\k \phi-\k\psi||_{2}\leq \sqrt{2k}
\]

\end{corollary}
\begin{proof}
Consider $L$ as a $\bZ$ lattice with generators $v_{1}, iv_{i},v_{2}, iv_{2},...v_{k},iv_{k}$. Since all of these are of unit length, the previous result applies.
\end{proof}
We now give the explicit relation between $\mathcal{N}$ and the approximate rank:
\begin{theorem}
Suppose $\chi_\delta(\ket{\phi})\leq k$. Then, we have that \[\N_{\delta+\delta_0}(\ket {\phi})\leq\frac {2k}{\delta_0^2}\]  
\end{theorem}
\begin{proof}
It suffices to prove the result for $\delta=0$. Let $\ket{\phi}=\sum_{i=1}^{k}c_i\ket{s_i}$, where $\ket{s_i}$ are stabilizer states lying on $BW_n$. Choose an arbitrary $c>0$. By Corollary \ref{cor_vec_latt_approx}, we know that there is a vector $\ket\psi$ on the $\mathbb Z[i]$ lattice generated by the $\ket{s_i}$s, which is $\sqrt{2k}$ close to $c\ket{\phi}$. Clearly, this lattice is a sublattice of $BW_n$, so we have a vector on $\ket{\psi}$ on the $BW_n$ such that:

    \[\norm{c\ket{\phi}-\ket{\psi}}_2^2\leq 2k \]
    \[\implies \norm{\ket\phi-\frac 1 c\ket{\psi}}_2^2\leq \frac {2k}{c^2}\]
    \[\implies \N_{\frac{\sqrt{2k}}{c}}(\ket{\phi})\leq {c^2}\]

    Fix $c=\sqrt{{2k}}/\delta_0$, and we get $\N_{\delta_0}\leq \frac{2k}{\delta_0^2}$ 
    
\end{proof}

\subsection{Relating $\N$ and $CS$-count for exact state preparation}

Consider the magic state $\ket{CS}=\frac 12\left({\ket {00}+\ket{01}+\ket{10}+i\ket{11}}\right)$. Using this state, we can apply the $CS$ gate with the protocol given in Figure \ref{fig_CS_magic} (Also see \cite{dyadic}). There, the label ``Clifford" refers to a Clifford circuit chosen based on the measurement outcome of the measured qubits. It is easy to see that for the measurement outcome $\ket{00}$, choosing the Clifford circuit to just be $I$ works, and this suffices for our use. Note that in this protocol, the measurements have uniform probability.
\begin{figure}
    \centering
\begin{tikzpicture}[
    line width=1.1pt,
    wire/.style={line width=1.1pt},
    leftbox/.style={
        draw,
        fill=white,
        line width=1.1pt,
        minimum width=0.55cm,   
        minimum height=1.60cm   
    },
    bigbox/.style={
        draw,
        fill=white,
        line width=1.1pt,
        minimum width=1.55cm,
        minimum height=1.55cm
    },
    meas/.style={
        draw,
        fill=white,
        line width=1.1pt,
        minimum width=0.65cm,
        minimum height=0.55cm
    },
    ctrl/.style={circle, fill=black, inner sep=0pt, minimum size=5pt},
    targ/.style={circle, draw, line width=1.1pt, inner sep=0pt, minimum size=17pt},
    every node/.style={font=\Large}
]

\node[anchor=east] at (-0.55,-0.5) {$|\psi\rangle$};

\draw[decorate, decoration={brace, amplitude=6pt}]
    (-0.30,-1.12) -- (-0.30,0.12);

\node[leftbox] (CSleft) at (1.45,-0.5) {CS};

\draw[wire] (0.05,0) -- (CSleft.west |- 0,0);
\draw[wire] (0.05,-1) -- (CSleft.west |- 0,-1);

\draw[wire] (CSleft.east |- 0,0) -- (2.35,0);
\draw[wire] (CSleft.east |- 0,-1) -- (2.35,-1);

\node at (2.95,-0.5) {$=$};

\def\xstart{5.35}
\def\xend{10.65}

\def\ytop{1.35}
\def\ysecond{0.55}
\def\ythird{-0.75}
\def\yfourth{-1.55}

\node[anchor=east] at (4.55,0.95) {$|\psi\rangle$};
\draw[decorate, decoration={brace, amplitude=6pt}]
    (4.80,0.35) -- (4.80,1.55);

\node[anchor=east] at (4.55,-1.15) {$|CS\rangle$};
\draw[decorate, decoration={brace, amplitude=6pt}]
    (4.80,-1.75) -- (4.80,-0.55);

\node[bigbox] (cliff) at (8.70,0.95) {Clifford};

\node[meas] (m1) at (8.70,\ythird) {};
\node[meas] (m2) at (8.70,\yfourth) {};

\draw[wire] (\xstart,\ytop) -- (cliff.west |- 0,\ytop);
\draw[wire] (cliff.east |- 0,\ytop) -- (\xend,\ytop);

\draw[wire] (\xstart,\ysecond) -- (cliff.west |- 0,\ysecond);
\draw[wire] (cliff.east |- 0,\ysecond) -- (\xend,\ysecond);

\draw[wire] (\xstart,\ythird) -- (m1.west);
\draw[wire] (m1.east) -- (\xend,\ythird);

\draw[wire] (\xstart,\yfourth) -- (m2.west);
\draw[wire] (m2.east) -- (\xend,\yfourth);

\node[ctrl] at (6.10,\ytop) {};
\draw[wire] (6.10,\ytop) -- (6.10,\ythird);
\node[targ] at (6.10,\ythird) {};
\draw[wire] (6.10,\ythird-0.16) -- (6.10,\ythird+0.16);
\draw[wire] (6.10-0.16,\ythird) -- (6.10+0.16,\ythird);

\node[ctrl] at (7.05,\ysecond) {};
\draw[wire] (7.05,\ysecond) -- (7.05,\yfourth);
\node[targ] at (7.05,\yfourth) {};
\draw[wire] (7.05,\yfourth-0.16) -- (7.05,\yfourth+0.16);
\draw[wire] (7.05-0.16,\yfourth) -- (7.05+0.16,\yfourth);

\draw[wire]
    (8.48,\ythird-0.08)
    .. controls (8.58,\ythird+0.10) and (8.77,\ythird+0.10)
    .. (8.92,\ythird-0.08);
\draw[->, line width=1pt]
    (8.73,\ythird+0.02) -- (8.92,\ythird+0.24);

\draw[wire]
    (8.48,\yfourth-0.08)
    .. controls (8.58,\yfourth+0.10) and (8.77,\yfourth+0.10)
    .. (8.92,\yfourth-0.08);
\draw[->, line width=1pt]
    (8.73,\yfourth+0.02) -- (8.92,\yfourth+0.24);

\draw[wire] (8.63,0.175) -- (8.63,\ythird+0.275);
\draw[wire] (8.77,0.175) -- (8.77,\ythird+0.275);

\draw[wire] (8.63,\ythird-0.275) -- (8.63,\yfourth+0.275);
\draw[wire] (8.77,\ythird-0.275) -- (8.77,\yfourth+0.275);
\end{tikzpicture}
\caption{$CS$-gate injection using the magic state $\k{CS}$. A similar protocol is given in \cite{dyadic}.}
\label{fig_CS_magic}
\end{figure}

Also, it is easy to verify that $\N(\ket{CS})=2$.
\begin{theorem}
Let $\ket{\phi}$ be an $n$-qubit state which can be prepared exactly with $m$ CS gates and Clifford operations without any intermediate measurement then:
\[
\mathcal{N}(\k\psi)\,{\big |}2^m
\]
If it does need intermediate measurements then we have the weaker result of \[\mathcal{N}(\k\psi)\leq 2^m \]
\end{theorem}
\begin{proof}

For the case when $\ket{\phi}$ does not require intermediate measurements,
we can also create $\ket{\phi}$ with $\k{CS}^{\otimes {m}}\otimes{\k{0}^{n}}$, we can create $\ket{\phi}$ exactly with Clifford gates, and the magic state protocol in Figure \ref{fig_CS_magic}, postselecting on all measurement outcomes being $\ket{0}$. By property (v) of Theorem \ref{thm_magic_props}, the result follows.

If there are indeed other measurements performed, which are possibly non-uniform, then the weaker result follows from property (vi) of Theorem \ref{thm_magic_props}.
\end{proof}

\section{Fidelity Amplification}
\label{sec_erroramplification}

In this section, we prove Theorem \ref{thm_error_amp}.

Define $\ket{H^\perp}=\sin\left(\frac \pi 8\right)\ket 0- \cos\left(\frac \pi 8\right)\ket 1$. Then, we know that $\ket {H}$ and $\ket{H^\perp}$ form an orthonormal basis for $\mathbb C^2$, and that $H\ket H=\ket H$, and $H\ket{H^\perp}=-\ket{H^\perp}$. We use this to ``cancel" out components orthogonal to $\ket{H}^{\otimes n}$ in any of its approximate stabilizer decomposition. First, we define a different notion of error that is easier to work with when states are not normalized. 

\begin{definition}
    Define relative error of $\ket{\phi'}$ for state $\ket{\phi}$ as \begin{equation}
    \frac{\Tr[(I-\ketbra \phi\phi )\,\ketbra{\phi'}{\phi'}]}{\Tr[\ketbra\phi\phi  \,\ketbra{\phi'}{\phi'}]}\label{eqn_error}
\end{equation}
\end{definition}
This can very easily be related to fidelity when $\norm{\ket{\phi'}}_2=1$, and the expression turns out to be $\frac{1-F(\ket\phi,\ket{\phi'})}{F(\ket\phi,\ket{\phi'})}$.

\begin{lemma}\label{lem_fid_amp}
    Given a rank $k$ decomposition of $\ket {H}^{\otimes n}$ with relative error $\epsilon$, there is a rank $2k$ decomposition with relative error $\frac \epsilon 2$.
\end{lemma}
\begin{proof}

For simplicity, we rewrite the form for relative error of any vector $\ket{\psi}$ with $\ket{H}^{\otimes n}$:
\begin{equation}
    \frac{\Tr[(I-(\ketbra HH)^{\otimes n} )\,\ketbra{\psi}{\psi}]}{\Tr[(\ketbra HH)^{\otimes n}  \,\ketbra{\psi}{\psi}]}\label{eqn_error_Htensorn}
\end{equation}
    We have a vector $\ket {\psi'}$ with relative error $\epsilon$ with respect to $\ket{H}^{\otimes{n}}$ with exact stabilizer rank $k$. Since relative error is invariant under scaling, without loss of generality $\ket{\psi'}$ can be taken to be of the following form:
     \[ \ket {\psi'} =  \ket {H}^{\otimes n} +\ket \phi  \]
 with $\Tr[\ketbra \phi \phi]=\epsilon$, and $\bra \phi H^{\otimes n}\rangle=0$. Note that $\ket{\psi'}$ need not be a normalized vector which allows us to get this decomposition, and plugging $\ket{\psi'}$ in place of $\ket{\psi}$ in Equation \ref{eqn_error_Htensorn} gives $\Tr[\ketbra \phi \phi]=\epsilon$.
 \\

Let $C$ denote a uniformly random chosen tensor product of $H$ and $I$. Then, we also have
\[ C\ket {\psi'} =  \ket {H}^{\otimes n} +C\ket \phi  ,\]
 with $C\ket{\phi}$ also orthogonal to $\ket{H}^{\otimes n}$. Thus,
\[ \frac{I+C}{2}\ket {\psi'} =  \ket {H}^{\otimes n} +\frac {I+C}{2}\ket \phi ,\]

with $\frac {I+C}{2}\ket \phi$ orthogonal to $\ket H^{\otimes n}$. Define the $LHS$ above to be $\ket{\psi''}$. Note that $\chi(\ket{\psi''})\leq 2\chi(\ket {\psi'})\leq 2k$.
Also, note that $\Tr[\ketbra{\psi'}{\psi'}\, (\ketbra H H)^{\otimes n}]=\Tr[\ketbra{\psi''}{\psi''}\, (\ketbra H H)^{\otimes n}]$. Thus, the denominator of Equation \ref{eqn_error_Htensorn} is the same for both $\ket{\psi'}$ and $\ket{\psi''}$.

Now, what we expect to happen is that $\frac{I+C}2 \ket {\phi}$ will have a lot of cancellations if written out in the $\ket H, \ket {H^{\perp}}$ basis. This is because, for each such basis element, except $\ket H ^{\otimes n}$ (which by orthogonality has coefficient $0$ in the expansion of $\ket \phi$), it is an eigenvector of $\frac {I+C}2$, with the corresponding eigenvalue being $1$ or $0$ with equal probability. We also know that for any vector $\ket x$, we have that $\braket{x|x}$ is the sum of squares of the coefficients of these basis elements. Thus, we have that for $ \ket {\phi'}=\frac {I+C}2 \ket \phi$, $\mathbb E[\braket{\phi'|\phi'}]=\frac 1 2 \braket {\phi|\phi}=\frac \epsilon 2$, where the randomness is over the choices of $C$. Therefore, there exists a $C$ for which $\braket{\phi'|\phi'}$ is at most $\frac \epsilon 2$.
This gives us our result.
\end{proof}
Repeating this $\log \alpha$ times, we get the following theorem.
\begin{theorem}
    Given a stabilizer decomposition of $\ket {H}^{\otimes n}$ with relative error $\epsilon$ and rank $k$, there exists another stabilizer decomposition with rank $O(\alpha k)$ and relative error $\frac {\epsilon}{\alpha}$. 
\end{theorem}

The original motive to find such a result was to assume existence of low rank decompositions with polynomially small error, given one for constant error. However, this also gives us a way to compose low rank approximate stabilizer decompositions together.
\begin{corollary}
    Given a $\delta$ relative error rank $k$ stabilizer decomposition for $\ket {H}^{\otimes m}$, there is a rank $  O((k(1+\delta))^n)$ constant relative error stabilizer decomposition for $\ket{H}^{\otimes mn}$.
\end{corollary}
\begin{proof}
    Suppose $\ket{\phi}$ has relative error $\delta$ with $\ket{H}^{\otimes m}$ and has stabilizer rank $k$. 

    Then, $\ket{\phi}^{\otimes n}$ has stabilizer rank $k^n$, and the relative error with respect to $\ket H ^{\otimes m n}$ is $(1+\delta)^n-1$. This is because, in Equation \ref{eqn_error}, for normalised states, the numerator and denominator sum to 1, and the denominator is multiplicative with tensor products of the states. Then, we get the result by applying Theorem \ref{thm_error_amp} with $\alpha=(1+\delta)^n-1$
\end{proof}

Note that if we apply the above result to $\ket{0}^{\otimes n}$, we retrieve back the $O(2^{0.23n})$ rank approximate stabilizer decomposition of $\ket{H}^{\otimes n}$ from \cite{Bravyi_2016}. In fact, our technique can be thought of as generalization for their construction.

\subsection{Relation to Lattice Approximations}\label{sec_latticeapprox}

Here, we discuss the relation of $BW_n$ lattice approximation of quantum states, where we say that a lattice vector $v$ on $BW_n$ approximates $\ket{\psi}$ if for some $c\in \mathbb C^*$, $\norm{cv-\ket{\psi}}_2$ is small. This is consistent with the notion of approximation used in Definition \ref{def_approx_bw_norm}.

In the proof of Lemma \ref{lem_fid_amp}, we are summing up $I \ket {\psi'}$ and $C\ket{\psi'}$ where $C$ is a uniformly random $n$-fold tensor product of $H$ and $I$.\footnote{We can ignore the factor of 1/2 as we effectively ignore scalar multiplication while computing the Barnes wall norm $\mathcal N$} Now, by Fact $\ref{fact_barnes_wall_automorphism}$, we know that whenever $C$ has even number of $H$s, $C$ also maps $BW_n$ to itself . Therefore, in such instances, if $\ket{\psi'}$ lies on $BW_n$, then so does $(I+C)\ket{\psi'}$. Therefore, if we restrict to $C$ having an even number of $H$s, the proof of Lemma \ref{lem_fid_amp} also acts as a way to jump from one lattice approximation of $\ket{H}^{\otimes n}$ to another.

It is easily seen that if $C$ is chosen uniformly at random from the set of $n$-fold tensor products of $H$ and $I$ with an even number of $I$'s, then the probability of cancelling out components orthogonal to $\ket{H}^{\otimes n}$ in the expansion of $\ket{\psi'}$, in the basis of tensor products of $\ket{H},\ket{H^{\perp}}$, is $\frac 1 2$ for each of the components except for the component along $\ket{H^{\perp}}^{\otimes n}$. Because of this final caveat, we can not in general say that the new vector on average has half the relative error with respect to $\ket{H^{\otimes n}}$, what we can say is the the relative error is halved with respect to the state $\ket{\psi'''}$, which is the normalised projection of $\ket{\psi'}$ onto the space spanned by $\{\ket{H}^{\otimes n},\ket{H^{\perp}}^{\otimes n}\}$. Also note that $\ket{\psi'''}$ does not change after sequencial application of this procedure, i.e, the projection of $(I+C)\ket{\psi'}$ on to $\mathrm{span}\{\ket{H}^{\otimes n},\ket{H^{\perp}}^{\otimes n}\}$ is also aligned along $\ket{\psi'''}$. Therefore, for the special case when $\ket{\psi'''}$ is a very good approximation for $\ket{H}^{\otimes n}$, this modified fidelity amplification procedure still works quite well.

To see why this special case is important, notice that for the setup in Corollary \ref{cor_decomp_compose}, if $\ket{\psi}$ is a good approximation for $\ket{H}^{\otimes m}$, then $\ket{\psi}^{\otimes n}$ will naturally have a very small component along $\ket{H^{\perp}}^{\otimes mn}$, relative to the component along $\ket{H}^{\otimes mn}$.

In particular, the component of the lattice vector $\ket{0}^{\otimes n}$ on to $\ket{H^{\perp}}^{\otimes n}$ is exponentially small relative to it's component along $\ket{H}^{\otimes n}$. Thus, $O(2^{0.23n})$-rank approximate decomposition for $\ket H^{\otimes n}$ obtained from Corolllary \ref{cor_decomp_compose} or from \cite{Bravyi_2016} can be specialized as a $BW_n$ vector approximation for $\ket{H}^{\otimes n}$. Moreover, it can be checked that on average, the square of the length of the lattice vector obtained (i.e, its Barnes wall norm $\mathcal N$) scales linearly in the rank. This shows that best known approximations for $\ket{H}^{\otimes n}$ are effectively lattice approximations which asymptotically match the bound in Theorem \ref{thm_monotone_vs_rank}.

\begin{lemma}[Informal]
Restricting the randomized $O(2^{0.23n})$-rank decomposition for $\ket{H}^{\otimes n}$ obtained by applying the Fidelility amplification procedure $0.23n$ times to only lattice approximations leads to lattice vector approximations for $\ket{H}^{\otimes n}$ with average squared-length (i.e, the Barnes wall norm $\mathcal N$) $O(2^{0.23n})$, 
\end{lemma}
\begin{proof}[Proof (Sketch)]
    We split up the problem in to the problem of finding the square of lengths of components along the orthonormal basis formed by tensor products of $\ket{H}$ and $\ket{H^{\perp}}$. This suffices as the square of the lengths of the components sum up to the square of the length of the vector itself.
    
    Suppose $\ket{\psi}$ is the random vector on $BW_n$ obtained at the end (after applying 0.23n iterations of the fidelity amplification procedure).

    Now, $\ket{\psi}$ decomposes as sums of tensor products of $\ket{0}$ and $\ket{+}$, with $O(2^{0.23n})$ terms. From this, we can see that the component along $\ket{H}^{\otimes n}$ is $O(2^{0.23n}\cdot \cos(\frac \pi 8))=O(2^{0.115n})$, and therefore the length-squared of the projection along $\ket{H}^{\otimes n}$ is $O(2^{0.23n})$.

    The length-squared of the component along $\ket{H^{\perp}}^{\otimes n}$ is negligible compared to the above. 

    For the component orthogonal to the space spanned by $\ket{H}^{\otimes n}$ and $\ket{H^{\perp}}^{\otimes n}$, we note that a basis for this subspace is the set of tensor products of $\ket{H}$ and $\ket{H^{\perp}}$, with $\ket{H}^{\otimes n}$ and $\ket{H^\perp}^{\otimes n}$ removed.

    For each such basis element, the length squared of the component along $\ket{\psi}$ is quadrupled with probability 1/2, and sent to 0 otherwise. So, on average, the length squared is doubled.

    Hence, overall, the length squared component along the space orthogonal to $\ket{H}^{\otimes n}$ and $\ket{H^{\perp}}^{\otimes n}$ is on average $O(2^{0.23n})$ times the initial value , i.e, the length squared of the component of $\ket{0}^{\otimes n}$ on to this space, which is clearly atmost 1. Therefore, the length squared of the lattice vector $\ket{\psi}$ on average is $O(2^{0.23n})$.

\end{proof}

\section{Exponential Exact Stabilizer Rank for a Product State}\label{sec_product_state_exist}
In this section, we prove the existence of pure states with maximal ranks. The argument for openness and density of these states, i.e, the proof of Theorem \ref{thm_product_state_dense}, is pushed to Section \ref{sec_product_state_dense}, which is a slight tightening of the argument presented in this section. Note that for both the arguments in this section and Section \ref{sec_product_state_dense}, we do not use any properties of the stabilizer states other than that there are only finitely many such states, and the argument below easily extends to decompositions with respect to any finite collection of states. 

Consider the space $\mathbb{C}^{2^{n}}$. Any subspace of it is called a $V_{k}$ space, if it can be written as the space spanned by some $k$ linearly independent stabilizer states $\{\ket{s_{i}}\}_{i=1}^{k}$. A subspace is called a $W_{k}$ subspace if it can be written as a non-zero intersection of $V_k$ subspaces.

Denote the set of $V_k$ subspaces as $\mathcal V_k$ and those of $W_k$ subspaces as $\mathcal W_k$. Note that both $\mathcal W_k$ and $\mathcal V_k$ are finite.
\begin{observation}
    For any $W_k$ subspace $B$, 

    \[B=\bigcap_{A\in \mathcal V_k, B\subseteq A} A\]
\end{observation}
\begin{proof}
    We know that 
     \[B=\bigcap_{A\in S} A\]

     for some subcollection $S$ of $\mathcal V_k$ (by definition of a $W_k$ subspace). Now, clearly, for each such $A$, we have $B\subseteq A$. Now, for any $A\notin S$ for which $B\subseteq A$, we can add $A$ to $S$ and the above expression will still be true, as it is equivalent to taking an intersection with $A$ on both sides.

\end{proof}

\begin{observation}
    For any $v\in \mathbb C^{2^n}$, the minimal $W_k$ subspace containing $v$ (ordered by containment) is 

    \[\bigcap_{A\in \mathcal V_k, v\in A} A \]

    Denote this $W_k$ subspace by $\mathcal W_{k,v}$.
\end{observation}
\begin{proof}
    Clearly, the above space is indeed a $W_k$ subspace containing $v$. Call this $B$ Now, any $W_k$ space containing $v$ can be written as 
    \[\bigcap_{A\in S} A \]
    where $S$ is some collection of $V_k$ spaces containing $v$. Written in this form, it is immediate that $B$ is contained in all such $W_k$ spaces.
\end{proof}

\begin{theorem}\label{thm_product_state_exist}
    There is a product state for which the exact stabilizer rank is $2^n$. 
\end{theorem}
\begin{proof}
    Let $k$ be the maximum possible exact stabilizer rank for $n$ qubit product states. Then, we know that all product states lie in corresponding $V_k$ spaces. 

    Let $\ket \phi$ be a product state such that $\mathcal W_{k,\ket \phi}$ is maximal, i.e, there is no other $\ket \psi$ for which $\mathcal W_{k,\phi} \subsetneq \mathcal W_{k,\psi}$.

    Consider $\mathcal B$ to be a small neighborhood around $\ket \phi$ (in the product states), such that for any $A\in \mathcal V_k, \ket{\phi}\notin A$, we have $\mathcal B\in A^C$. This is possible because there are only finitely many such $A$'s, and each of them has non zero distance from $\ket \phi$, so we can just take the ball around $\ket \phi$ with small enough radius.

        \begin{lemma}
            $\mathcal B\subseteq \mathcal W_{k,\ket \phi}$.
        \end{lemma}
        \begin{proof}
            Pick any $b\in \mathcal B$. We know that $\mathcal W_{k,b}=\bigcap_{A\in \mathcal V_k,b\in A} A\supseteq \bigcap_{A\in \mathcal V_k,\phi \in A}= \mathcal W_{k,\ket \phi} $ by the choice of $\mathcal B$. However, by the choice of $\ket \phi$, we know that this containment cannot be strict. Therefore, these two subspaces are equal. The lemma follows directly from this.
        \end{proof}
            \begin{lemma}
                If a subspace $V$ of $\mathbb C^{2^n}$ contains an open set $\mathcal B$ (in the product states), then $V=\mathbb C^{2^n}$.
            \end{lemma}
            \begin{proof}
                WLOG, suppose the open set contains $\ket 0^{\otimes n}$. We can assume this as we can shift the product state by application of product unitary operators.

                Consider the state $s(\epsilon)=\cos \epsilon \ket 0 +\sin \epsilon \ket 1$.
                Then, for a small enough $\epsilon$, all $n$-fold tensor products of $\ket 0$ and $s(\epsilon)$ are going to lie in $\mathcal B$. This is because $\mathcal B$ will contain small enough open balls around $\ket 0^{\otimes n}$.

                However, these $2^n$ different tensor products form a basis for $\mathbb C^{2^n}$.
            \end{proof} 

            Combining both of these results tell us that $\mathcal W_{k,\phi}=\mathbb C^{2^n}$. Since the dimension of any $W_k$ space is at most $k$ (as they are intersections of $V_k$ spaces), we get $k\geq 2^n$.
\end{proof}

The idea of using these $V_k$ spaces originates from their use in Mehraban and Tahmasbi's proof of the quadratic lower bound\cite{Saeed}. However, it turns out these are the very objects that even Lovitz and Steffan \cite{Lovitz2022newtechniques} used in their proof of these results, but instead treated them as algebraic structures called varieties. Their proof technique also directly gives us the result that the set of maximal rank states has full measure. However, as we directly make use of the notion of distances, we believe our techniques might be more suitable for working on the problem of approximate stabilizer ranks (see Problem \ref{open_prob_prod}).

\subsection{Density and Openness of Maximum Rank States}
\label{sec_product_state_dense}
In the proof Theorem \ref{thm_product_state_exist}, we assumed $\ket{\phi}$ to have maximal $\mathcal W_{k,\ket{\phi}}$ among all product states. However, looking at the argument carefully, we can see that it only has to be maximal in a neighborhood. We use this observation to extend the proof to the following result:
\begin{theorem}
The set of product states with stabilizer rank $2^n$ is open and dense in the set of product states.
\end{theorem}
\begin{proof}
    Fix $k=2^n-1$. For openness, we use a simple topological argument.  Observe that each $V_k$ space is closed in $\mathbb C^{2^n}$, and thus so is their union.

    Now, in $\mathbb C^{2^n}$, the set of vectors which have stabilizer rank exactly $2^n$ is $\mathbb C^{2^n}\setminus \bigcup_{A\in \mathcal V_k} A$, which is open in $\mathbb C^{2^n}$. Since the topology on the product states is the subspace topology when considering the product states as a subset of $\mathbb C^{2^n}$, the openness follows.

    For density, fix any $\ket{\phi}$ and $\epsilon$. Let the $\epsilon$ ball around $\ket {\phi}$ in the product states be $B_i$. If there any product state in $B_i$ which has stabilizer rank $2^n$, then we are done.
    
    Suppose this is not the case. Then, choose $\ket{\phi'}$ with the maximal $\mathcal W_{k,\ket{\phi'}}$ in $B_1$. Now, choose a ball $B_2$ centered at $\ket{\phi'}$ such that $B_2\subseteq B_1$. 
    
    From this, we have that all states in $B_2$ have stabilizer rank at most $k=2^n-1$, and that $\mathcal W_{k,\ket{\phi'}}$ is maximal in $B_2$. Then, by the argument in proof of Theorem \ref{thm_product_state_exist}, we have that $k=2^n$, which is a contradiction. 
\end{proof}

\section{Directions for Future work}

Here, we outline some open problems that arise naturally due to our work, which we hope will help towards solving the stabilizer rank problem:

First, in the proof of Theorem \ref{thm_extent_vs_rank}, we related the smallest eigenvalue of the Gram Matrix to the Stabilizer extent and the rank of the corresponding decomposition. Is it possible to use the structure of the lattice to bound more eigenvalues?

\begin{openproblem}
    In proof of Theorem \ref{thm_extent_vs_rank}, find more relations between eigenvalues and the extent.
\end{openproblem}

\noindent If it were possible to do this, it would strengthen the results in Theorem \ref{thm_extent_vs_rank}. Next, we believe trying to bound $\N_\delta(\ket{H}^{\otimes n})$ is a viable strategy to understand $\chi_{\delta}(\ket H^{\otimes n})$.

\begin{openproblem}
    Find bounds on $\N_{\delta}(\ket H^{\otimes n})$.
\end{openproblem}
\noindent In particular, lower bounds for this would imply corresponding lower bounds on $\chi_{\delta}(\k H^{\otimes n})$. 

\noindent Our Fidelity Amplification result, Theorem \ref{thm_error_amp}, trades off stabilizer rank and relative error. Is it possible to create other such trade offs, such as between rank and number of qubits? For example, could we prove that existence of low stabilizer rank decompositions of $\ket{H}^{\otimes n}$ imply even lower rank decompositions of $\ket{H}^{\otimes m}$ for $m<n$?
\begin{openproblem}
Find analogues to Fidelity Amplification for other trade offs.
\end{openproblem}

For the next problem, note that the existence of exact stabilizer rank $2^n$ product states directly implies that there is some function $\delta: \mathbb N \to (0,1)$ such that for each $n$, $\chi_{\delta(n)}(\ket\psi)=2^n$ for some product state $\ket{\psi}$. One can see this by a limiting argument. However, the limiting argument does not tell us what values of $\delta(n)$ work. So, a natural question would be to narrow this down.

\begin{openproblem}\label{open_prob_prod}
    Find examples of $\delta(n)$ for which $\chi_{\delta(n)}(\ket\psi)=2^{\Omega(n)}$ for some product state $\ket \psi$. 
\end{openproblem}

Note that if the above holds for any $\delta(n)=2^{-O(n)}$, then $\chi(\ket H^{\otimes n})$ will have to be super polynomial. This is because we can get exponentially small error for product state preparation with only polynomially many $T$ gates.

 Finally, it can be useful generalize the notions of the monotones $\N$ and $\N_\delta$ using the Barnes Wall lattices corresponding to $\mathbb Z[e^{i\pi/4}]$, as defined in \cite{kliuchnikov2024stabilizer}. This is because then there would be some vector proportional to $\ket{T}^{\otimes n}$ on the lattice, which is not the case for $BW_n$. Using that, we can get a corresponding lower bound on the $T$-counts of state preparation analogous to Theorem \ref{thm_monotone_CS_count}.

 Note that simply generalizing the definition of $\N$ and $\N_\delta$ to correspond to the usual euclidean lengths is not useful, as $\mathbb Z[e^{i\pi/4}]$ is dense in $\mathbb C$.

 To counteract this, instead of directly working with the $l_2$ norm of the vectors, we can work with some norm function $F$ applied to $\braket{\psi|\psi}$. The norm function defined in \cite{kliuchnikov2024stabilizer} is a candidate, but since it isn't multiplicative, most of the arguments in Section \ref{sec_magic} fail.

 Next, we can consider the field  norm for the extension $\mathbb Q(\sqrt 2)/\mathbb Q$. Since it can be negative, we may further take the absolute value. Using this norm, almost all the arguments in Section \ref{sec_magic} can be recreated, except, we have no knowledge about the minimal vectors with respect to this norm. It would still be the case that the stabilizer states have norm 1, but the arguments in \cite{kliuchnikov2024stabilizer} do not exclude the possibility of even smaller vectors with respect to this norm. In case there are indeed even smaller vectors, then this would lose any significance as a measure of magic. 

 \begin{openproblem}
     Find a natural generalization of $\N$ and $\N_\delta$ for the Barnes Wall lattice for $\mathbb Z[e^{i\pi/4}]$, such that the results in Section \ref{sec_magic} still apply.
 \end{openproblem}

Beyond $\mathbb Z[e^{i\pi/4}]$, most of the arguments fail as the higher cyclotomic fields are no longer Euclidean Domains.

\section{Acknowledgments}
ARK and PS are very grateful to David Gosset for suggesting the connection between Barnes Wall Lattices and Stabilizer Ranks and for subsequent related discussions. We also thank Graeme Smith, Benjamin Lovitz and Yash Totani for discussions. PS is supported by Mike and Ophelia Lazaridis Fellowship. ARK was supported in part by Canada’s NSERC, NTT Research, and the Royal Bank of Canada. IQC and the Perimeter Institute (PI) are supported in part by the Government of Canada through ISED and the Province of Ontario.

\section*{Author Contribution}
The authors jointly discussed all the technical results and wrote the paper collaboratively. LLMs were used to resolve typesetting and formatting problems and to generate Tikz code for Figure \ref{fig_CS_magic}. The outputs have been verified by the authors.

\bibliographystyle{quantum}
\bibliography{rank}
\end{document}